\documentclass[12pt]{article}
\usepackage{amsfonts,amsmath,latexsym,amssymb,amsthm}
\usepackage{graphicx}
\usepackage{natbib} 
\usepackage{url} 
\usepackage{algorithmic}
\usepackage{graphicx}
\usepackage{textcomp}
\usepackage{xcolor}

 \usepackage{subfigure}
  \usepackage{multirow}
\usepackage[utf8]{inputenc}

\usepackage[]{algorithm2e}

\usepackage{booktabs,tabularx}
\usepackage{dcolumn}

\newcommand{\blind}{0}

\newcounter{theoremcounter}
\newcounter{lemmacounter}

\newcounter{dummycounter}

\newcounter{corcounter}

\newcounter{emptycounter}
\newcounter{defcounter}

\newtheorem{theorem}[theoremcounter]{Theorem}

\newtheorem{lemma}[lemmacounter]{Lemma}

\newtheorem{corollary}[corcounter]{Corollary}
\newtheorem{remark}{Remark}
\newtheorem{definition}[defcounter]{Definition}

\numberwithin{equation}{section}
\numberwithin{lemmacounter}{section}
\numberwithin{propcounter}{section}
\numberwithin{corcounter}{section}
\numberwithin{conjcounter}{section}
\numberwithin{theoremcounter}{section}
\numberwithin{probcounter}{section}
\numberwithin{conditioncounter}{section}

\newcounter{eqncounter}

\numberwithin{equation}{eqncounter}

\newcommand{\model}{M,K}
\newcommand{\truemodel}{M_0,K_0}
\newcommand{\BICmodel}{M^*,K^*}
\newcommand{\RBICmodel}{\tilde{M},\tilde{K}}

\newcommand{\maxM}{M_{\max}}
\newcommand{\maxK}{K_{\max}}

\DeclareMathOperator*{\argmin}{arg\,min} 

 \usepackage{setspace}
\begin{document}




\if0\blind
{
  \title{\bf Robust Function-on-Function Regression}
  \author{Harjit Hullait\\
    STOR-i Centre for Doctoral Training, Lancaster University\\
    and \\
    David S. Leslie \\
    Department of Mathematics and Statistics, Lancaster University\\
    and \\
    Nicos G. Pavlidis \\
    Department of Management Science, Lancaster University\\
    and \\
    Steve King \\
    Rolls Royce PLC}
  \maketitle
} \fi

\if1\blind
{
  \bigskip
  \bigskip
  \bigskip
  \begin{center}
    {\LARGE\bf Title}
\end{center}
  \medskip
} \fi

\bigskip
\begin{abstract}


Functional linear regression is a widely used approach to model functional responses with respect to functional inputs. However, classical functional linear regression models can be severely affected by outliers. We therefore introduce a Fisher-consistent robust functional linear regression model that is able to effectively fit data in the presence of outliers. The model is built using robust functional principal component and least squares regression estimators. The performance of the functional linear regression model depends on the number of principal components used. We therefore introduce a consistent robust model selection procedure to choose the number of principal components. Our robust functional linear regression model can be used alongside an outlier detection procedure to effectively identify abnormal functional responses. A simulation study shows our method is able to effectively capture the regression behaviour in the presence of outliers, and is able to find the outliers with high accuracy. We demonstrate the usefulness of our method on jet engine sensor data. We identify outliers that would not be found if the functional responses were modelled independently of the functional input, or using non-robust methods. 

\end{abstract}

\noindent%
\vfill



\newpage
\section{Introduction}
\label{sec:intro}



Functional Linear Regression (FLR) in the function-on-function case \citep{tools} is a widely used technique for modelling functional responses with respect to functional inputs. The FLR model is able to capture complex dependency structures as it uses information across time \citep{Morris2015}. However classical FLR models can be severely affected by outliers as we will demonstrate via a simulation study in Section \ref{sec:simulations}. We therefore develop a robust FLR (RFLR) model, which is able to effectively fit the data in the presence of outliers. The model is built using the robust Functional Principal Component model by \cite{bali2011} and the multivariate Least Trimmed Squares (MLTS) estimator by \cite{Agullo2008}. The RFLR model can be used to identify abnormal functional responses, i.e. samples in which the functional behaviour between the predictor and response curves deviate from normal. 



Our study of robust FLR is motivated by the need to identify normal relationships in jet engine sensor data when we expect outliers to be present. The data is collected during Pass-Off tests, which are performed on an engine before deployment. In a Pass-Off test a human controller performs manoeuvres, which can be defined as various engine accelerations and decelerations starting and ending at a set idle speed. During the test, data is captured by sensors measuring engine speed, pressure, temperature and vibration in different parts of the engine. One of the key manoeuvres in a Pass-Off test is the Vibration Survey (VS). In this manoeuvre the engine is accelerated slowly to a certain speed then slowly decelerated. We have 199 VS datasets, which include the turbine pressure ratio (TPR) that measures the engine speed, and various temperature features including the turbine gas temperature (TGT). In Figure \ref{fig:wave_engine} we have plots of the TPR and TGT for the first 30 VS manoeuvres. To anonymise the data we have transformed the time index onto the interval $[0,1]$ and the sensor measurements to the range $[0,100]$. 
 
The VS manoeuvres are performed by a human controller, which causes variability in the TPR curves as can be seen in Figure \ref{fig:wave_engine}. This variability will naturally affect the TGT curves and may mask the unusual behaviour produced by the engine. Our results in Section \ref{sec:engine} support this claim, and show that using a direct outlier detection on the engine temperature curves fails to identify meaningful outliers. Instead this approach picks up curves produced by unusual TPR speed profiles. We therefore require a method of detecting outliers in the presence of the controller induced variability. We expect that the relationship between the engine speed and engine temperature for different VS manoeuvres should be the same irrespective of the way the manoeuvre is performed. For example given a certain engine acceleration we would expect a certain temperature response. If however the response differs from expectation this could be indicative of an engine issue. In Section \ref{sec:Outlier_Detection} we will show how RFLR can be used for outlier detection, which is later applied to the jet engine data in Section \ref{sec:engine}, to identify abnormal behaviour.

The paper is organised as follows. In Section \ref{sec: Classical Functional Data Analysis} we outline the classical FLR model. In Section \ref{sec:Robust_FLR}, we will outline robust Functional Data Analysis (FDA) techniques to obtain a robust FLR model. We also introduce a robust model selection procedure. In Section \ref{sec:Asymptotic_Results} we prove consistency results for the robust FLR model and the robust model selection procedure. In Section \ref{sec:Outlier_Detection}, we describe an outlier detection method, which acts on the residuals of the robust FLR model. In Section \ref{sec:simulations} we perform a simulation study to illustrate the model fitting and outlier detection capabilities of the robust model. In Section \ref{sec:engine} we apply the robust model on the engine data and highlight unusual observations that can not be detected by using outlier detection directly on the temperature curves. Finally in Section \ref{sec:conc} we provide a conclusion. 


\begin{figure}
\centering     

\subfigure{\label{fig:wave_TPR}\includegraphics[width=8cm]{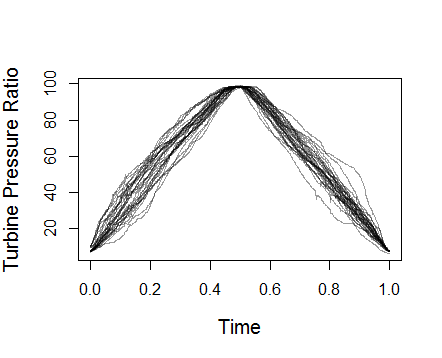}}
\subfigure{\label{fig:wave_tgt}\includegraphics[width=8cm]{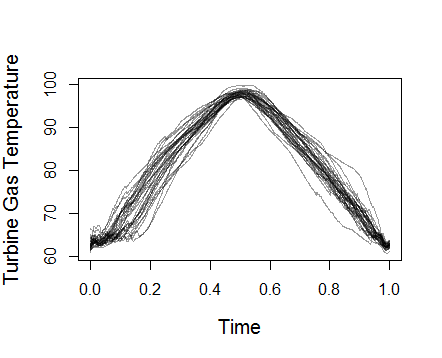}}
  \caption{Plots of the first 30 TPR and TGT time series.}
\label{fig:wave_engine}
\end{figure}





\section{Classical Functional Data Analysis}\label{sec: Classical Functional Data Analysis}

In this section we give a brief summary of the FDA tools that we will later apply in our model. In the following sections we will use the vector space $L^2(I)$ which is the Hilbert space of square integrable functions on the compact interval $I$ with the inner product $\langle f,g \rangle =\int_{I}f(t)g(t)dt$ for functions $f,g\in L^2(I)$. 

We will define $X(t), Y(t)$ to be univariate stochastic processes defined on $I$, with mean functions $\mu^X(t)$ and $\mu^Y(t)$, and covariance functions $C_X(s,t)=cov\{X(s),X(t)\}$ and $C_Y(s,t)=cov\{Y(s),Y(t)\}$ for all $s,t\in I$. We shall define $x(t)=[x_{1}(t),...,x_{n}(t)]$ and $y(t)=[y_{1}(t),...,y_{n}(t)]$ to be $n$ independent and identically distributed realisations of $X(t)$ and $Y(t)$ respectively. 

In practice we observe $x_i(t)$ and $y_i(t)$ at discrete time points. We shall assume for simplicity of exposition that observations are made at equally spaced time points $t_1,...,t_T$. We will outline Functional Linear Regression and Functional Principal Component Analysis with respect to the underlying functions $x(t), y(t)$. In Section \ref{sec:Model_select} we need to use the discretely observed data to define a suitable model selection criterion.

\subsection{Functional Linear Regression}\label{sec:FLR}


In this section we will introduce the classical FLR model \citep{tools}. In FLR we model the relationship between predictor $x_i(t)$ and response $y_i(t)$ as: 
 
\begin{equation}\label{eq:linear}
y_{i}(t)=\alpha(t) + \int_{I}x_{i}(s)\beta(s,t)ds +\epsilon_i(t),
\end{equation}
where $\alpha(t)$ is the intercept function, $\beta(s,t)$ is the regression function and $\epsilon_i(t)$ is the error process. For a fixed $t$, we can think of $\beta(s,t)$ as the relative weight placed on $x_i(s)$ to predict $y_i(t)$. As in \citet{CHIOU2016} we will assume the mean functions $\mu^X(t)=0$ and $\mu^Y(t)=0$ which thereby means $\alpha(t)=0$. This is a reasonable assumption as in practice we can calculate the mean functions $\mu^X(t)$ and $\mu^Y(t)$ efficiently for dense data and then pre-process the data by subtracting $\mu^X(t)$ and $\mu^Y(t)$ from the observed curves.  


FLR in the function-on-function case can be modelled parametrically \citep{yaofang2005, CHIOU2016} or nonparametrically \citep{FERRATY2012, Ivanescu2015, Scheipl2015}. We use a parametric approach which models the regression matrix in terms of pre-defined basis functions.  

We will represent $x_i(t)$ and $y_i(t)$ in terms of $(M,K)$ pre-chosen basis functions $\phi^X_j(t),\phi^Y_j(t)$ respectively:

\begin{equation*}
x_i(t)=\sum_{m=1}^{M} z_{im}\phi^X_m(t) \text{ and } y_i(t)=\sum_{k=1}^{K} w_{ik}\phi^Y_k(t).
\end{equation*}
where $z_{im}, w_{ik}\in\mathbb{R}$.
 
We define $\phi^X(t)=[\phi^X_1(t),...,\phi^X_M(t)]$, $\phi^Y(s)=[\phi^Y_1(s),...,\phi^Y_K(s)]$, $z_i=[z_{i1},...,z_{iM}]$ and $w_i=[w_{i1},...,w_{iK}]$. We will then model the regression surface using a double basis expansion \citep{Ramsay2005}:
  
\begin{equation}\label{eq:double_basis}
\beta(s,t)=\sum_{m=1}^{M}\sum_{k=1}^{K}b_{mk}\phi^X_{m}(s)\phi^Y_{k}(t)=\phi^X(s)^{T}B\phi^Y(t),
\end{equation}
for an $M\times K$ regression matrix $B$. We can then write:
 
\begin{equation}\label{eq:basis_expand}
y_i(t)=z_iB\phi^Y(t)+\epsilon_i(t).
\end{equation}


Letting $\epsilon_i(t)=q_i\phi^Y(t)$ we can reduce Equation \eqref{eq:basis_expand} to:

\begin{equation}\label{eq:LS}
w_i=z_iB+q_i.
\end{equation}
This parametrisation of the residual function is also used by \cite{CHIOU2016}. We can then estimate $B$ using standard multivariate regression methods typically assuming Gaussian $q_i$.
 

\subsection{Functional Principal Component Analysis}\label{sec:FPCA}


In this section we describe Functional Principal Component Analysis (FPCA), which we will use to build data-driven basis functions $\phi^X(t)$ and $\phi^Y(t)$ for $x_i(t)$ and $y_i(t)$, respectively. These basis functions give effective, low-dimensional representations and will be used in the Functional Linear Regression model described in Section \ref{sec:FLR}. 

Functional Principal Component Analysis (FPCA) is a method of finding dominant modes of variance for functional data. These dominant modes of variance are called the Functional Principal Components (FPCs). FPCA is also used as a dimensionality reduction tool, as a set of observed curves can be effectively approximated by a linear combination of a small set of FPCs. 

The FPCs, $\phi^X_{m}(t)$ for $m=1,2,...$, are the eigenfunctions of the covariance function $C_X(s,t)$ with eigenvalues $\lambda^X_m$. Note that the eigenfunctions are ordered by the respective eigenvalues. The Karhunen-Lo\'{e}ve theorem \citep{Shang2014} shows that $x_i(t)$ can be decomposed as $x_{i}(t)=\sum_{m=1}^{\infty}z_{im}\phi^X_{m}(t)$ where the principal component score $z_{im}=\int_{I}x_{i}(t)\phi^X_{m}(t)dt$. The scores $z_{im}$ are realisations from a random variable $\xi^X_m$. 

We can define the $M$-truncation as
\begin{equation}\label{eigenapprox}
\hat{x}^{M}_{i}(t)=\sum_{m=1}^{M}z_{im}\phi^X_{m}(t),
\end{equation}
which gives the minimal residual error:

\begin{equation}\label{fit}
\frac{1}{n}\sum_{i=1}^{n}||x_{i}-\hat{x}^{M}_{i}||^2=\frac{1}{n}\sum_{i=1}^{n}\int_{I} [x_{i}(t)-\hat{x}^{M}_{i}(t)]^2 dt.
\end{equation}
To choose $M$ we will use an information criterion outlined in Section \ref{sec:Model_select}. An analogous procedure is used to find $K$ eigenfunctions $\phi_k^Y(t)$ for $y(t)$. 
 

 \subsection{Bayesian Information Criterion for FLR}\label{sec:Model_select}


In the FLR model described in Section \ref{sec:FLR} we need to choose terms $M$ and $K$. Typically $M$ and $K$ are chosen independently \citep{yaofang2005},  however the estimation of $\beta(s,t)$ also depends on $M$ and $K$ and this should be incorporated into the estimation of these terms. In this section we formulate a Bayesian Information Criterion (BIC) to determine the basis size $M$ and $K$, similarly to \cite{matsui2017}.  

A component of the BIC is the log likelihood, often expressed as a squared error term. It is tempting to use the squared error resulting from Equation \eqref{eq:LS}. However the objective is to fit the data $y_i$ so we should use a likelihood of this data instead of a squared error term of basis coefficients.

We have a set of models $J=\{(\model)| M=1,...,\maxM, K=1,...,\maxK\}$, where $\maxM$ and $\maxK$ are pre-set maximum number of FPCs that will be considered in the model. Let vector $\vec{y}_i$ be the values of $y_i(t)$ evaluated at discrete time points: $\vec{y}_i=[y_i(t_1),...,y_i(t_{T})]$. Let $z_i^{(M)}$ be the first $M$ principal scores of $x_i(t)$ with respect to the FPCs $\phi^X(t)$ and let $\phi^{(K)}$ be the matrix with $(k,i)$ entry $\phi^Y_k(t_i)$. We assume there exists a true model $(M_0,K_0)$ with associated $M_0\times K_0$ matrix $B^{M_0, K_0}$ such that 

 \begin{equation}\label{eq:basis_expand_discrete}
 \vec{y}_i = (z^{(M_0)}_i)^T B^{\truemodel} \phi^{(K_0)} +\epsilon_i,
\end{equation} 
where the error $\epsilon_i=[\epsilon_i(t_1),...,\epsilon_i(t_T)]$ is assumed for simplicity to be sampled from $N(0,v^2I_T)$, where $I_T$ is the identity matrix of size $T$. 

For Model $(\model)$ we define $\theta^{\model}=(B^{\model},v^{\model})$ and the prediction $\hat{y}_i^{\model}=(z^{(M)}_i)^T B^{\model} \phi^{(K)}$. We want to identify this true model $(\truemodel)$, which we can use to obtain consistent estimates of $\theta^{\truemodel}$. 

 For Model $(\model)$ we can define the likelihood for sample $i$ as

\begin{equation}
f(\vec{y}_i|\theta^{\model})=\frac{1}{(2\pi)^{\frac{T}{2}}(v^{\model})^{T}}\exp\left\{-\frac{[\vec{y}_i-\hat{y}_i^{\model}]^T[\vec{y}_i-\hat{y}_i^{\model}]}{2(v^{\model})^2}\right\},
\end{equation} 
and the log-likelihood $l(\theta^{\model})=\sum_{i=1}^n \log(f(\vec{y}_i|\theta^{\model}))$. As in \citet{eilers1996} 
 \begin{equation}\label{eq:BIC}
 BIC_n(\model) = -2l(\hat{\theta}^{\model})+w(\model)\log(n)
 \end{equation}
 where $\hat{\theta}^{\model}$ is the maximum likelihood estimator and the penalty $\omega(\model)=MK+1$, in which $MK$ is the number of free parameters in the model and the 1 comes from $v$. We will denote $(\BICmodel)_n=\argmin_{(\model)\in J}BIC_n(\model)$, which is dependent on the sample size $n$. 
 
To summarise, we estimate the FPCs for $X$ and $Y$ and solve the FLR model for different models $(\model)$. We then choose model $(\BICmodel)_n$ that minimises the BIC criterion. The robust equivalent of this procedure is given in Algorithm \ref{alg:Robust_FLR}.  
 
\section{Robust Functional Linear Regression}\label{sec:Robust_FLR}

In Section \ref{sec: Classical Functional Data Analysis} we have defined the FLR model and have outlined the use of FPCA bases to estimate parameters of the model. In this section we will introduce robust versions of the FDA techniques outlined in Section \ref{sec: Classical Functional Data Analysis}. This will allow us to fit a normality model even in the presence of outliers. We shall also propose a robust BIC procedure for model selection.   

We will replace classical FPCA with robust FPCA estimates by \cite{bali2011} which ensure that outliers do not unduly affect the FPCA estimates. Note that FPCA minimises the residual error given in \eqref{fit}. To obtain robust FPCA estimates \cite{bali2011} minimise a robust scale estimator, using a projection pursuit approach, which iteratively performs a weighted least squares till the estimators stabilise.   

Analogous to \eqref{eigenapprox}, the robust FPCs $\tilde{\phi}^X_m(t)$ ($m=1,...,M$) and $\tilde{\phi}^Y_k(t)$ ($k=1,...,K$) are orthonormal functions such that 

\begin{equation*}\label{eigenapprox_rob}
x_{i}(t)\approx\sum_{m=1}^{M}\tilde{z}_{im}\tilde{\phi}^X_{m}(t),\quad y_{i}(t)\approx\sum_{k=1}^{K}\tilde{w}_{ik}\tilde{\phi}^Y_{k}(t),
\end{equation*} 
are good approximations for $x_i(t)$ and $y_i(t)$.

We define $\tilde{y}_i(t)=\tilde{w}_i\tilde{\phi}^Y(t)$ and assume as in \eqref{eq:LS} that $\epsilon_i=\tilde{q}_i\tilde{\phi}^Y(s)$. We can now write the robust counterpart of \eqref{eq:LS} as

\begin{equation}\label{eq:LS_rob}
\tilde{w}_i=\tilde{z}_i\tilde{B} + \tilde{q}_i.
\end{equation}

To obtain a robust estimate of the regression matrix $\tilde{B}$, we will use the Multivariate Least Trimmed Squares (MLTS) estimator by \cite{Agullo2008}, to mitigate the affect of outliers with respect to the regression relationship. Given $\alpha\in[0,1]$ we can define $r=[\alpha n]$ as the $\alpha$ proportion of samples rounded to the nearest integer, and the set $\mathcal{S}=\{S\subset\{1,...,n\}, |S|=r\}$. The objective of MLTS is to find a subset $S$ such that

\begin{equation*}
S=\argmin_{S\in\mathcal{S}}\sum_{i\in S}||\tilde{w}_i-\tilde{z}_i\tilde{B}||^2.
\end{equation*}
This is robust as outliers will not be in the subset by definition so shall not affect the model estimation. We will choose a subset of size $r= [0.8n]$.



\subsection{Robust Bayesian Information Criterion for FLR}\label{sec:RBIC}

The BIC model selection method is known to be non-robust \citep{Machado1993}. In particular outliers can significantly affect the loglikelihood estimation. We therefore outline a robust BIC (RBIC) model, which, similar to MLTS, maximises over a subset of samples $S$. RBIC can therefore give good model selection performance in the presence of outliers. 

We will define $\tilde{\theta}^{\model}=(\tilde{B}^{\model},\tilde{v}^{\model})$ as robust estimated parameters for model $(\model)$ and the robust prediction $\tilde{y}_i^{\model}=(\tilde{z}^{(M)}_i)^T \tilde{B}^{\model} \tilde{\phi}^{(K)}$. We define the trimmed likelihood for model $(\model)$ and set $S$ as  

\begin{equation}\label{eq:rob_like}
\tilde{l}(\tilde{\theta}^{\model},S)=\sum_{i\in S}\left(\frac{[\vec{y}_i-\tilde{y}_i^{\model}]^T[\vec{y}_i-\tilde{y}_i^{\model}]}{(\tilde{v}^{\model})^2}\right)+rT\log(2\pi)+2rT\log(\tilde{v}^{\model}).
\end{equation}


We will define $S^{\model} = \argmin_{S\in\mathcal{S}}\tilde{l}(\tilde{\theta}^{\model},S)$, where $\mathcal{S}=\{S\subset\{1,...,n\}, |S|=r\}$ for $r= [0.8n]$. Then 
\begin{align}\label{eq:RBIC}
RBIC_n(\model)  & = -2\min_{S\in\mathcal{S}}\tilde{l}(\tilde{\theta}^{\model},S)+\omega(\model)\log(r) \\
& = -2\tilde{l}(\tilde{\theta}^{\model},S^{\model})+w(\model)\log(r)
\end{align}

We will denote $(\RBICmodel)_n=\argmin_{(\model)\in J}RBIC_n(\model)$, and we will assume that this minimum is unique. 


In Algorithm \ref{alg:Robust_FLR} we outline the calculation of the robust FLR model, which incorporates the RBIC procedure. In the algorithm we estimate the model for different values of $(\model)$ and choose the model with the minimum RBIC value. We consider $M=1,…,\maxM$ and $l=1,...,\maxK$ where $\maxM,\maxK$ are chosen to ensure that $99.99\%$ of the variance in the raw data is captured.




\begin{algorithm}
 \KwData{Let $(x_i,y_i)$ be mean-corrected time series of length $T$ for $i=1,...,n$.}
 \begin{algorithmic} 
 \STATE 1. Estimate $\{\tilde{\phi}^X_1(t),...,\tilde{\phi}^X_{\maxM}(t)\}$, $\{\tilde{\phi}^Y_1(t),...,\tilde{\phi}^Y_{\maxK}(t)\}$ \citep{bali2011}.
  \FOR {$M=1,...,\maxM$}
    \FOR {$K=1,...,\maxK$}
      \STATE Estimate the regression matrix $B^{\model}$ using MLTS \citep{Agullo2008}.
      \STATE Obtain the $RBIC_n(\model)=\argmin_{(\model)\in J}RBIC_n(\model)$ value using \eqref{eq:RBIC} 
    \ENDFOR 
\ENDFOR
\STATE 2. Select model $(\RBICmodel)_n$.
\RETURN Regression matrix $\tilde{B}$ from model $(\RBICmodel)_n$ and $\{\tilde{\phi}^X_1(t),...,\tilde{\phi}^X_{\tilde{M}}(t)\}$, $\{\tilde{\phi}^Y_1(t),...,\tilde{\phi}^Y_{\tilde{K}}(t)\}$.
\end{algorithmic}
 \caption{Robust FLR procedure}
 \label{alg:Robust_FLR}
\end{algorithm}

\section{Asymptotic Results}
\label{sec:Asymptotic_Results}


In Section \ref{sec:Robust_FLR} we proposed a Robust FLR model for the function-on-function problem. A minimum criteria for a good model is consistency, i.e. that given an ideal scenario of unlimited data that the estimator will be equal or arbitrarily close to the truth. In this section we shall prove consistency and Fisher-consistency for the robust FLR model. We shall follow a similar approach to \cite{Kalogridis2018} who developed a robust FLR model for the scalar-on-function problem. We shall also prove the consistency of the RBIC model selection method outlined in Section \ref{sec:Model_select}.  



\begin{definition}\label{def:consistency}
Let $X_1,X_2,..., X_n$ be a sequence of real-valued random variables. An estimator $T_n:=T(X_1,X_2,..., X_n)$ of a parameter $\theta$ is said to be  (asymptotically) \textbf{consistent} if for all $\epsilon>0$

\begin{equation*}
\lim_{n\rightarrow\infty}P(|T_n-\theta|>\epsilon)=0.
\end{equation*} 
\end{definition}

\begin{definition}\label{def:Fisher}
Let $X_1,X_2,..., X_n$ be a sequence of real-valued random variables with an associated cumulative distribution function $F_{\theta}$, which depends on an unknown parameter $\theta$. Let the estimator $T_n:=T(F_n)$ of a parameter $\theta$, be a function of the empirical distribution function $F_n$. We say this estimator is \textbf{Fisher-consistent} for the parameter $\theta$ if 

\begin{equation*}
T(F_{\theta})=\theta
\end{equation*}
\end{definition}

\begin{remark}\label{remark:consist}
Fisher consistency is equivalent to (asymptotic) consistency if the empirical distribution function $F_n$ converges pointwise to the true distribution function $F_\theta$. This can be shown to be the case for iid real multivariate random variables using the Glivenko-Cantelli theorem \citep{pollard2012}.

\end{remark} 


\subsection{Consistency of the Robust FLR}


To prove Fisher-consistency we need to define appropriate probability measures on the predictor $X(t)$, response $Y(t)$ and the residual $\epsilon(t)$. We will then define conditions by which the robust FPCA and MLTS regression are Fisher-consistent, which will then ensure the Fisher-consistency of $\tilde{\beta}(s,t)$. We shall also prove consistency of $\tilde{\beta}(s,t)$ using Remark \ref{remark:consist}. Following the ideas set by  \cite{Kalogridis2018}, we make 6 assumptions:

\begin{itemize}
\item[(C1)] $X$ has a finite-dimensional Karhunen-Lo\'{e}ve decomposition, i.e $\lambda^X_{m}=0$ for $m>M_0$.
\item[(C2)] $Y$ has a finite-dimensional Karhunen-Lo\'{e}ve decomposition, so $\lambda^Y_{k}=0$ for $k>K_0$.
\item[(C3)] The residual $\epsilon(t)=\tilde{q}\tilde{\phi}^Y(t)$ where $\tilde{q}$ is a Gaussian random variable with mean $0$ and covariance matrix $\Sigma$. 
\item[(C4)] $\beta(s,t)$ lies in a linear subspace spanned by $\{\tilde{\phi}^X_m\}_{m=1}^{M_0}$ and $\{\tilde{\phi}^Y_k\}_{k=1}^{K_0}$.
\item[(C5)] The random variables $\{\tilde{\xi}^X_j\}_{j=1}^{M_0}$ are absolutely continuous and have joint density $g_1(x)$ satisfying $g_1(x)=h_1(||x||_E)$ for $x\in\mathbb{R}^{M_0}$ and some measurable function $h_1:\mathbb{R}\rightarrow\mathbb{R}_{+}$.
\item[(C6)] The random variables $\{\tilde{\xi}^Y_j\}_{j=1}^{K_0}$ are absolutely continuous and have joint density $g_2(y)$ satisfying $g_2(y)=h_2(||y||_E)$ for $y\in\mathbb{R}^{K_0}$ and some measurable function $h_2:\mathbb{R}\rightarrow\mathbb{R}_{+}$.

\end{itemize}
We define $||\cdot||_E$ as the Euclidean norm.

Let $P_X$ be the image measure of $X$ i.e. $P_X(U)=P(X\in U)$ for a Borel set $U$, and likewise for $P_Y$. We can define the cumulative distribution functions 

\begin{align*}
&F_X(a_1,...,a_{M_0}):=P_X(\tilde{\xi}^X_1\leq a_1,...,\tilde{\xi}^X_{M_0}\leq a_{M_0}), \\
&F_Y(b_1,...,b_{K_0}):=P_Y(\tilde{\xi}^Y_1\leq b_1,...,\tilde{\xi}^Y_{K_0}\leq b_{K_0}). 
\end{align*}




Let $F_{\epsilon}$ denote the distribution function of $\epsilon(t)$, which can be defined in the same way as $F_X$ and $F_Y$. We can write the functional of the robust estimator $\tilde{\beta}(s,t)$ as:

\begin{equation}
\tilde{\beta}(F_{\epsilon},F_X,F_Y)(s,t)=\sum_{k=1}^{K_0}\sum_{m=1}^{M_0}\hat{B}_{km}(F_{\epsilon},F_X,F_Y)\tilde{\phi}^X_{m}(F_X)(s)\tilde{\phi}^Y_{k}(F_Y)(t).
\end{equation}

The functional is Fisher-consistent if $\tilde{\beta}(F_{\epsilon},F_X,F_Y)(s,t)=\beta(s,t)$ for $s,t\in I$, which in turn follows from $\tilde{B}_{km}(F_{\epsilon},F_X,F_Y)=B_{km}$, $\hat{\phi}^Y_{k}(F_Y)(t)=\phi^Y_{k}(t)$ and $\hat{\phi}^X_{m}(F_X)(t)=\phi^X_{m}(s)$. Conditions C1-C4 are to ensure the FLR problem can be defined by a finite number of terms. \cite{Kalogridis2018} show that Conditions C5 and C6 are sufficient for the Fisher-consistency of the robust FPCA estimators by \cite{bali2011}.

\begin{lemma}\label{lemma:Fisher}
Assume C1-C6 holds then $\tilde{\beta}(F_{\epsilon},F_X,F_Y)(s,t)$ is Fisher-consistent. 
\end{lemma}

\begin{proof}

Conditions C1-C2 and C5-C6 ensure Fisher-consistency of the robust FPCA estimators as shown by \cite{bali2011}, so $\tilde{\phi}^Y(F_Y)(t)=\phi^Y(t)$ and $\tilde{\phi}^X(F_X)(t)=\phi^X(t)$. By conditions C1-C2 we can write 

\begin{equation*}
Y(t)=c\tilde{\phi}^Y(F_Y)(t), \quad X(t)= Z\tilde{\phi}^X(F_X)(t)
\end{equation*}

Then 

\begin{align*}
\int_{I}X(s)\tilde{\beta}(F_{\epsilon},F_X,F_Y)(s,t)ds &= \int_{I}Z\tilde{\phi}^X(F_X)(s)\tilde{\phi}^X(F_X)(s)^{T}\tilde{B}(F_{\epsilon},F_X,F_Y)\tilde{\phi}^Y(F_Y)(t)ds \text{ using C4}\\
& = Z\tilde{B}(F_{\epsilon},F_X,F_Y)\tilde{\phi}^Y(F_Y)(t). 
\end{align*}
Using condition C3 we can write $\epsilon(t)=\tilde{q}\tilde{\phi}^Y(t)$ therefore 

\begin{equation*}
Z\tilde{B}(F_{\epsilon},F_X,F_Y)\tilde{\phi}^Y(F_Y)(t) + \epsilon(t) = Z\tilde{B}(F_{\epsilon},F_X,F_Y)\tilde{\phi}^Y(F_Y)(t) + \tilde{q}\tilde{\phi}^Y(F_Y)(t),
\end{equation*}
multiplying by $\tilde{\phi}^Y(F_Y)(t)$ and integrating over $t$ we obtain 

\begin{equation*}
Z\tilde{B}(F_{\epsilon},F_X,F_Y)+\tilde{q}.
\end{equation*}
\cite{Agullo2008} show that Condition C3 implies the MLTS estimator is Fisher-consistent so $\tilde{B}(F_{\epsilon},F_X,F_Y)=B$. Therefore $\tilde{\beta}(F_{\epsilon},F_X,F_Y)(s,t)ds=\beta(s,t)$. 
 
\end{proof}

\begin{corollary}\label{corollary:beta}
If $\{x_1(t),y_1(t)\},...,\{x_n(t),y_n(t)\}$ are iid samples with cumulative distribution function $(F_X,F_Y)$. Then, assuming C1-C6 holds, $\tilde{\beta}(s,t)$ is consistent. 
\end{corollary}
Note that $x_i(t)$ and $y_i(t)$ are defined on a finite number of eigenfunctions, so are defined by finite score vectors. Therefore Corollary \ref{corollary:beta} follows from Lemma \ref{lemma:Fisher} and Remark \ref{remark:consist}, which states almost sure convergence of the empirical distribution for iid multivariate random variables. In this case Fisher-consistency is equivalent to consistency.   

 

\subsection{Consistency of RBIC}


We defined RBIC for the FLR problem in Section \ref{sec:RBIC}. In this section we will prove consistency of RBIC for the FLR problem. We will assume there is a true model, which we previously defined as $(\truemodel)$. We can then define overspecified and underspecified models in reference to this true model. We make some assumptions on the behaviour of the likelihood for these two model classes to prove consistency. We also denoted $(\RBICmodel)_n=\min_{(\model)\in J}RBIC_n(\model)$, which we will assume is unique. 

We will split the candidate models in $J$ into two sets, one is the overspecified models that include the true model $J_+=\{(\model)\in J|M\geq M_0 \text{ and } K\geq K_0\}$ and underspecified models $J_-=J_+^c\cap J$. Recall that $r=[\alpha n]$ for some $\alpha\in (0,1)$, and the likelihood $\tilde{l}$ in \eqref{eq:rob_like} depends on $r$ terms. 


\textbf{Assumption 1}
For $(\model)\in J_-$, there exists some $\varepsilon^{\model}>0$ such that
\begin{equation*}
\lim_{n\rightarrow\infty}P\left[\frac{1}{r}(\tilde{l}(\tilde{\theta}^{\truemodel},S^{\truemodel})-\tilde{l}(\tilde{\theta}^{\model},S^{\model}) )>\varepsilon^{\model}\right]=1.
\end{equation*}  
This is a reasonable assumption as the underspecified models should give a poorer fit to $y_i$ than the true model. 

\textbf{Assumption 2}
For $(\model)\in J_+$, there exists some $\gamma^{\model}>0$ such that 

\begin{equation*}
\lim_{n\rightarrow\infty}P\left[\tilde{l}(\tilde{\theta}^{\model},S^{\model})-\tilde{l}(\tilde{\theta}^{\truemodel},S^{\truemodel})>\gamma^{\model}\right]=0.
\end{equation*}
This assumption states that the difference in the trimmed loglikelihood is less than a finite $\gamma$. The likelihood for the overspecified models and the true model should be close, given the true model is contained within the overspecified models, so the difference in the penalty terms will outweigh the difference in the likelihoods for large enough $n$. 

Note that in Assumption 1 we consider the average difference between the log-likelihoods, whereas in Assumption 2 we look at the total difference.


 \begin{theorem}\label{theorem:consistency_RBIC}
 Given Assumptions 1 and 2 hold, and the true model $(\truemodel)\in J$ then $(\RBICmodel)_n$ is a consistent estimator of $(\truemodel)$.
 \end{theorem}

 \begin{proof}
 
 For $j\in J_-$, we will show 
 
 \begin{equation}\label{eq:check}
 \lim_{n\rightarrow\infty}P(\{RBIC_n(\model)-RBIC_n(\truemodel)\}>0)=1.
 \end{equation}
 By definition we can show that:
 \begin{align*}
  & \lim_{n\rightarrow\infty} P\left(RBIC_n(\model)-RBIC_n(\truemodel)>0\right)\\
  & = \lim_{n\rightarrow\infty} P\left(-2\left(\frac{\tilde{l}(\tilde{\theta}^{\model},S^{\model})-\tilde{l}(\tilde{\theta}^{\truemodel},S^{\truemodel})}{r}\right) > -\frac{(\omega(\model)-\omega(\truemodel))\log(r)}{r}\right).
\end{align*}
We will label $H_r=-2\left(\frac{\tilde{l}(\tilde{\theta}^{\model},S^{\model})-\tilde{l}(\tilde{\theta}^{\truemodel},S^{\truemodel})}{r}\right)$ and $G_r=\frac{(\omega(\model)-\omega(\truemodel))\log(r)}{r}$. Using $\varepsilon^{\model}$ from Assumption 1, we can see that $-G_r<2\varepsilon^{\model}$ for sufficiently large $r$. Using this and Assumption 1 we can show 

\begin{equation*}
\lim_{n\rightarrow\infty}P(H_r>-G_r)\geq \lim_{n\rightarrow\infty}P(H_r>2\varepsilon^{\model})=1.
\end{equation*}
Therefore $\lim_{n\rightarrow\infty} P\left(RBIC_n(\model)-RBIC_n(\truemodel)>0\right) = 1$ for $(\model)\in J_-$. 
 
  For $(\model)\in J_+\backslash\{(\truemodel)\}$, we know that $\frac{1}{2}(\omega(\model)-\omega(\truemodel))\log(r)>0$ and is monotonically increasing. Therefore there exists $N$ such that for $r\geq N$ 
  \begin{equation}\label{eq:u}
  \frac{1}{2}(w(\model)-w(\truemodel))\log(r)>\gamma^{\model}.
\end{equation}   
We can show that 
  
  \begin{align*}
  & \lim_{n\rightarrow\infty}P\left(RBIC_n(\model)-RBIC_n(\truemodel)<0\right)\\
  &   = \lim_{n\rightarrow\infty}P\left([\tilde{l}(\tilde{\theta}^{\model},S^{\model})-\tilde{l}(\tilde{\theta}^{\truemodel},S^{\truemodel})]>\frac{1}{2}(\omega(\model)-\omega(\truemodel))\log(r)\right)\\
  &   \leq \lim_{n\rightarrow\infty}P\left([\tilde{l}(\tilde{\theta}^{\model},S^{\model})-\tilde{l}(\tilde{\theta}^{\truemodel},S^{\truemodel})]>\gamma^{\model}\right)=0 \text{ by Assumption 2}.
  \end{align*} 
 
      
 \end{proof}
Note that BIC is a special case of RBIC where $r=n$, so is also consistent by Theorem \ref{theorem:consistency_RBIC}.

\section{Outlier Detection}\label{sec:Outlier_Detection}

There is a rich literature of outlier detection methods for functional data (FD). There are functional depth based methods such as the thresholding approach by \cite{Bande2008} and the functional boxplot by \cite{Sun2011}. Alternatively we can use methods based on outlyingness measures such as \cite{ArribasGil2014}, and \cite{Dai2018}. For multivariate FD there exist outlier detection methods such as \cite{Rousseeuw2018} and \cite{Hubert2015}. These methods do not model the dependency between the functional response and functional input, and may therefore miss important outliers. This will be shown in the simulation study in Section \ref{sec:simulations}. RFLR can model this dependency structure, which can improve the detection of outliers. We therefore suggest an outlier detection algorithm which uses RFLR to model the dependency structure. Using residuals from the model we can apply standard outlier detection approaches. The outliers in the residuals will be samples that are not well explained by the model which fits the majority of the curves.


The RFLR model produces estimates of the responses $\tilde{y}_i(t)=\tilde{z}_i\tilde{B}\tilde{\phi}^Y(t)$ for $i=1,...,n$. For an outlier we expect the residual curve $r_i(t) = y_i(t)-\tilde{y}_i(t)$ to deviate in behaviour from the other residuals. Traditionally, we would use the integrated square error to identify outliers. However using a functional depth approach \citep{Bande2008} is more effective in identifying outliers in functional data, in particular shape outliers that are not unusual if viewed at each time point but are abnormal across the entire trajectory. The approach assigns a depth value to samples $r_i(t)$. Samples with small depth values lie far away from the other samples. 

We will use the $h$-modal depth \citep{Cuevas2007} to rank samples $r_i$. For a given kernel $G_h$ (typically Gaussian with bandwidth $h$), the $h$-modal depth of $r_i$ with respect to $r=\{r_1,...,r_n\}$ is given by: 

\begin{equation}\label{eq:modal}
D(r_i|r,h) = E(G_h(||r_i-r||)) \approx \frac{1}{n}\sum_{l=1}^n G\left(\frac{||r_i-r_l||}{h}\right).
\end{equation} 
The $h$-modal depth has two useful properties. First, it uses a distance metric therefore samples further away from the centre will be given a smaller depth value. Second, in the case of multiple ``normal'' types behaviour, the $h$-modal depth works effectively as it doesn't assume there is one centre. 


In the algorithm we need to choose the bandwidth $h$ and a threshold $C$ to identify outliers. The bandwidth $h$ is taken to be the 15th percentile of the empirical distribution of $\{||r_i-r_j||, i,j=1,...,n\}$ \citep{Bande2008}. The threshold $C$ is chosen such that $P(D(r_i|r,h)\leq C) = \delta$, where $\delta$ is a pre-chosen percentile. To estimate the threshold $C$ they use a bootstrapping approach, which estimates a value of $C$ for different random sets of samples and then aggregates these estimates. We describe the outlier detection algorithm in Algorithm \ref{alg:OD}. 



 
 \begin{algorithm}
 \KwData{Centred curves $\{x_i(t),y_i(t)\}$ for $i=1,...,n$ and percentile $\delta$.}
 \begin{algorithmic}
  \STATE 1. Use Algorithm \ref{alg:Robust_FLR} to obtain $\tilde{\phi}^Y_k(t)$, $\tilde{z}_m$ and $\tilde{B}$.
  \STATE 2. Calculate residual curves $r_i(t)$.
 \STATE 3. Estimate bandwidth $h$.
 \STATE 4. For each $r_i(t)$ calculate $D(r_i|r,h)$.
 \STATE 5. Estimate $C$ for given percentile $\delta$.      
 \STATE 6. If $D(r_i|r,h)<C$ sample $i$ is an outlier.
\end{algorithmic}
 \caption{Outlier Detection using robust FLR.}
 \label{alg:OD}
\end{algorithm}

\section{Simulation Study}
\label{sec:simulations}

In this section we will provide a simulation study to investigate the finite sample properties of RBIC and robust FLR (RFLR) in comparison to BIC and classical FLR (CFLR). In the simulation study we will generate data using a FLR process and corrupt a certain number of samples, which will be the outliers. The outliers have been designed to be undetectable, if the response curves are considered independently of the predictor curves. Therefore standard functional data outlier detection algorithms such as those discussed in Section \ref{sec:Outlier_Detection} will perform poorly. 



The main motivation for the RFLR model is to obtain good model fitting in the presence of outliers. In this simulation study we compare the fitting error (FE) given in \eqref{eq:FE_non}, for the non-outlier samples using the robust model, which uses RFLR and RBIC with the classical approach using CFLR and BIC. We define the indicator variable $u_i=1$ if sample $i$ is an outlier and 0 otherwise. Letting $\hat{y}_i(t)$ be the estimation of $y_i(t)$ and given that proportion $a$ of the samples have been contaminated then FE is given by:

\begin{equation}\label{eq:FE_non}
FE = \frac{1}{(1-a)n}\sum_{i=1}^n(1-u_i)||y_i-\hat{y}_i||^2.
\end{equation}

Next we compare the outlier detection capabilities of robust and classical approaches using the receiver operating characteristic (ROC) curve to determine the sensitivity/specificity trade-off for different thresholds. If we have perfect outlier detection for all thresholds then the area under the curve (AUC) of the ROC curve would be 1. We can therefore use the AUC value as a measure of outlier detection accuracy regardless of threshold.   



FPCA is performed by taking the principal components of a 200 cubic B-spline representation of each of the predictor and response curves \citep{Ramsay2005}. The robust FPCA approach outlined in Section \ref{sec:Robust_FLR} is performed using the CR algorithm proposed by \cite{Croux1996} on the same B-spline coefficients. The MLTS estimator is calculated using the heuristic given by \cite{Agullo2008} using different trimming proportions $(1-\alpha)$ for $\alpha\in [0,1]$.

\subsection{Scenarios}



We will generate samples $x(t)$ using a FPCA based model with mean function $\mu_X(t)= -10(t-0.5)^2 + 2$ for $t\in[0,1]$ and eigenfunctions:

\begin{equation*}
  \phi^X_1=\sqrt{2}\sin(\pi t),\hspace{0.5em} \phi^X_2=\sqrt{2}\sin(7\pi t),\hspace{0.5em} \phi^X_3=\sqrt{2}\cos(7\pi t).
\end{equation*}   
The principal scores are sampled from Gaussian distributions with mean 0 and variances 40, 10 and 1 for the eigenfunctions respectively. Note that we do not create any outliers in the FPCA decompositions of the predictor curves.  We generate 400 predictor curves $x_1(t),...,x_{400}(t)$, which are observed at $T=500$ equidistant points in the interval $[0,1]$. 

The samples $y(t)$ will have eigenfunctions: 

\begin{equation*}
  \phi^Y_1=\sqrt{2}\sin(12\pi t),\hspace{0.5em}  \phi^Y_2=\sqrt{2}\sin(5\pi t),\hspace{0.5em}  \phi^Y_3=\sqrt{2}\cos(2\pi t),
\end{equation*}   
and mean function $\mu_Y(t)=60\exp(-(t-1)^2)$. We will generate $\beta(s,t)=\phi^X(s)^{T}B\phi^Y(t)$ where $B$ will have random entries between $[-3,3]$. We generate non-outlier curves:

\begin{equation*}
y_i(t)=\mu_Y(t)+\int_I\beta(s,t)(x_i(s)-\mu_X(s))ds +\epsilon_i(t),
\end{equation*}
where the residual function $\epsilon_i(t)=q_i\phi^Y(t)+d_i$ where $q_i$ and $d_i$ are sampled iid from $N(0,0.1)$. We will consider three cases when the proportion of outliers are $a=0.1,0.2$ and $0.3$.  


In \textbf{Scenario 1} outliers will be generated by replacing $B$ with $B_1=B+R$ where $R$ has random entries sampled from $N(0,0.5)$ giving $\beta_1(s,t)=\phi^X(s)^{T}B_1\phi^Y(t)$. Outliers $y'_i(t)$ are given by

\begin{equation*}
y'_i(t)=\mu_Y(t)+\int_I\beta_1(s,t)(x_i(s)-\mu_X(s))ds +\epsilon_i(t).
\end{equation*}

In \textbf{Scenario 2} we generate outliers by adding a random B-spline function $p(t)$ defined on an interval of length $1/10$. Letting $\beta_2(s,t)=\phi^X(s)^{T}B_2[\phi^Y(t), p(t)]$, for $3\times 4$ matrix $B_2=[B, l]$ for $l\sim N(2,1)$, then the outliers $y''_i(t)$ are given by

\begin{equation*}
y''_i(t)=\mu_Y(t)+\int_I\beta_2(s,t)(x_i(s)-\mu_X(s))ds +\epsilon_i(t).
\end{equation*} 
Note that the outliers in Scenario 1 affect the regression function across the entire interval whereas the outliers in Scenario 2 only affect a small interval of the curves.

In Figure \ref{fig:scenario_plots} we have a plot of the predictor curves $x_i(t)$ and response curves $y_i(t)$ with outliers from Scenario 1 and Scenario 2. The figure shows the outliers are masked by the variability in the curves and therefore cannot by identified using standard outlier detection algorithms. To make the outliers clearer we have plotted the residuals of the response curves using the true regression function and mean functions. In the bottom row of Figure \ref{fig:scenario_plots} we can see that the outliers in Scenario 2 are localised to a fixed interval whereas in Scenario 1 the outliers affect the response curve at all time points.  

The RFLR model depends on the proportion of trimming $\alpha$. To investigate the effect of the trimming we will consider trimming proportions $\alpha=0.1,0.2$ and $0.3$. We shall also investigate the performance using BIC and RBIC with fixed trimmed sample size of $r=[0.8n]$. 


We sample 400 predictor and response curve datasets and generate classical and robust models to calculate the average FE \eqref{eq:FE_non}. In Tables \ref{table:scenario1_BIC} and \ref{table:scenario2_BIC} we present the results for Scenario 1 and 2 respectively. The CFLR model gives a smaller FE value in the case of no-outliers $a=0$, however the robust model still gives good model fits. If we compare the FE using BIC and RBIC, we can see that BIC gives better model choices when $a=0$. This is due to BIC using all the data and in particular using samples in the tails of the distribution. In the presence of outliers the robust model outperforms the classical model, and as expected the difference in FE increases as the number of outliers increases. We should also note that RBIC is giving better model choices than BIC when outliers are present. Next, we can see using trimming proportion $\alpha=0.1$ we obtain significantly large FE values when $a=0.3$. However the FE values for $\alpha=0.2$ and $0.3$ are very similar in the case of $a=0.3$. The outliers generated can have different sizes, therefore in the $\alpha=0.2$ robust model only small outliers are present, which only affect the model fitting slightly . 

In Figure \ref{fig:ROC_plots} we have two ROC curves generated for one of the repetitions in Scenario 1 and 2 in which we have contaminated $20\%$ of the samples. In both scenarios the robust model outperforms the classical model. We also deploy the approach of \citet{Bande2008} to the response curves, disregarding the predictor curves (henceforth called the Direct approach). The ROC curves show that the robust and classical models are more effective than the Direct method in identifying the outliers in Scenario 1 and 2. By only using the specificity and sensitivity for a fixed threshold a lot of information is being lost, therefore a better comparison would be the area under the curve (AUC). Using the AUC metric we can understand the model outlier detection capabilities overall, in particular how well are the outliers separated from the other samples. We have taken the average AUC values over the 100 iterations performed for Scenario 1, which are shown in Table \ref{table:scenario1_AUC}. We have considered the average AUC values for trimming levels $\alpha=0.1,0.2$ and $0.3$. The robust models give larger AUC values than the classical model. However the different trimming levels do not seem to have a significant effect on the AUC values. In Scenario 2 we have the results in Table \ref{table:scenario2_AUC}. The same patterns appear as in Scenario 1 except the AUC values are notably smaller. This is to be expected given the outliers in Scenario 2 are defined on a small time interval.  

  

\begin{figure}
\centering     
\subfigure[$x_i(t)$]{\label{fig:predictor1}\includegraphics[width=6cm]{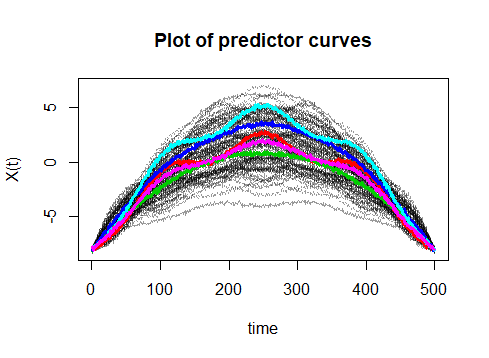}}
\subfigure[$x_i(t)$]{\label{fig:predictor2}\includegraphics[width=6cm]{Predictor_curves_col.png}}
\subfigure[$y^{(1)}_i(t)$]{\label{fig:response1}\includegraphics[width=6cm]{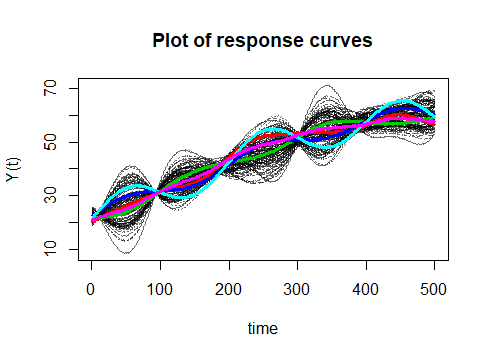}}
\subfigure[$y^{(2)}_i(t)$]{\label{fig:response2}\includegraphics[width=6cm]{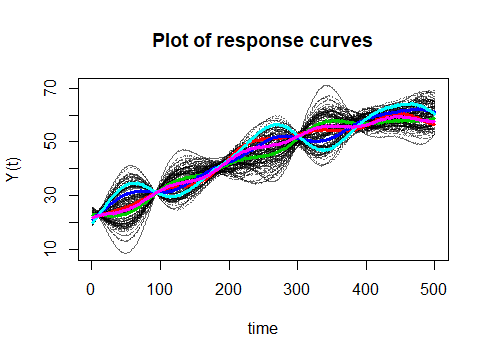}}
\subfigure[$r^{(1)}_i(t)$]{\label{fig:residual1}\includegraphics[width=6cm]{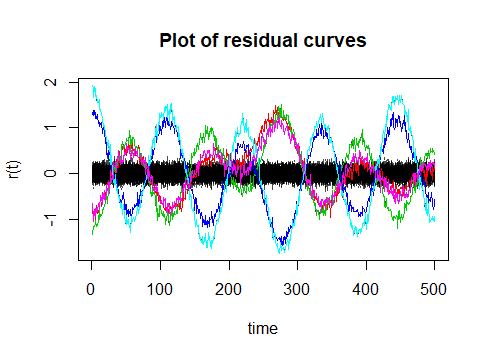}}
\subfigure[$r^{(2)}_i(t)$]{\label{fig:residual2}\includegraphics[width=6cm]{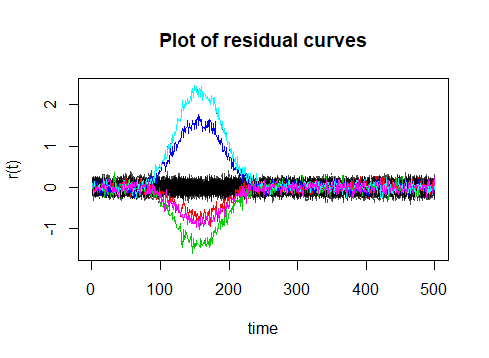}}
\caption{\textit{Left:} Plots of the predictor curves $x_i(t)$, response curves $y^{(1)}_i(t)$ and residuals curves $r^{(1)}_i(t)$ for Scenario 1. \textit{Right:} Plots of the predictor curves $x_i(t)$, response curves $y^{(2)}_i(t)$ and residuals curves $r^{(2)}_i(t)$ for Scenario 2. The residual curves are generated using the true regression function and mean functions. In each scenario there are 5 outliers each in a distinctive colour. The predictors curves $x_i(t)$ are identical for both scenarios, and the response curves look very similar due to mean and functional components masking the outliers. However the residuals are clearly distinctive. }\label{fig:scenario_plots}
\end{figure}

\begin{table}
\begin{center}
  \caption{ \quad Average fitting errors (FE) for 100 replications for Scenario 1, using classic FPCA and robust FPCA with different amount of trimming in the MLTS estimator and using models selected by BIC and RBIC.}
        \begin{tabular}{lllllll}\toprule
            & Trim & Model & a=0 & a=0.1 & a=0.2 & a=0.3 \\\midrule
            Classic & $\alpha=0.0$ & BIC & \textbf{5.326} & 18.441 & 48.771 & 101.320 \\\midrule
            Robust 
                  &  $\alpha=0.1$ & BIC & 8.283 & 14.166 & 21.118 & 33.907 \\
                  &  $\alpha=0.1$ & RBIC & 9.285 & \textbf{9.179} & 10.674 & 28.393 \\
                  &  $\alpha=0.2$ & BIC & 8.288 & 14.178 & 15.750 & 16.623 \\
                  &  $\alpha=0.2$ & RBIC & 9.292 & 9.207 & \textbf{9.535} & 13.436 \\
                  &  $\alpha=0.3$ & BIC & 8.294 & 14.199 & 15.815 & 16.518 \\
                  &  $\alpha=0.3$ & RBIC & 9.301 & 9.214 & 9.544 & \textbf{12.334} \\
            \bottomrule
        \end{tabular}
        \label{table:scenario1_BIC}
		\end{center}
    \end{table} 

\begin{table}
\begin{center}
  \caption{ \quad Average fitting errors (FE) for 100 replications for Scenario 2, using classic FPCA and robust FPCA with different amount of trimming in the MLTS estimator and using models selected by BIC and RBIC.}
        \begin{tabular}{lllllll}\toprule
            & Trim & Model & a=0 & a=0.1 & a=0.2 & a=0.3 \\\midrule
            Classic & $\alpha=0.0$ & BIC & \textbf{5.326} & 17.252 & 48.906 & 85.063 \\\midrule
            Robust 
                  &  $\alpha=0.1$ & BIC & 8.283 & 15.242 & 21.524 & 28.758 \\
                  &  $\alpha=0.1$ & RBIC & 9.285 & \textbf{9.074} & 9.919 & 18.546 \\
                  &  $\alpha=0.2$ & BIC & 8.288 & 16.745 & 20.652 & 21.928 \\
                  &  $\alpha=0.2$ & RBIC & 9.292 & 9.191 & \textbf{8.997} & 13.628 \\
                  &  $\alpha=0.3$ & BIC & 8.294 & 16.808 & 20.695 & 21.750 \\
                  &  $\alpha=0.3$ & RBIC & 9.301 & 9.233 & 9.018 &  \textbf{11.439} \\
            \bottomrule
        \end{tabular}
        \label{table:scenario2_BIC}
		\end{center}
    \end{table} 
    
    \begin{figure}
\centering     
\subfigure[Scenario 1]{\label{fig:ROC1}\includegraphics[width=7cm]{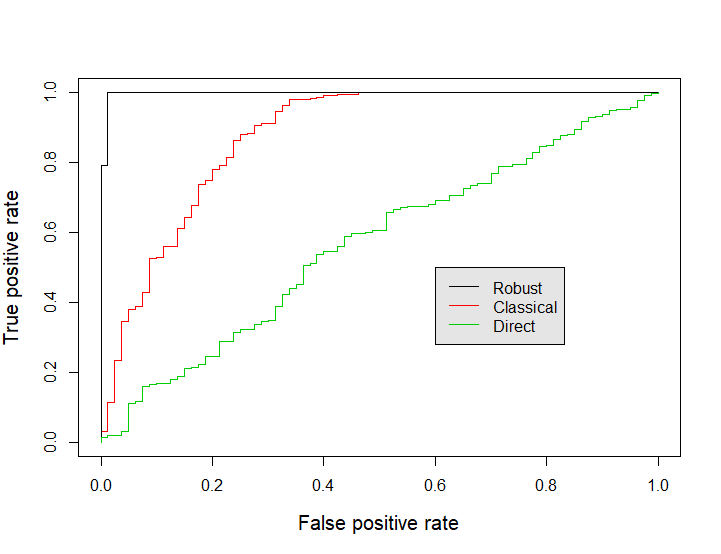}} 
\subfigure[Scenario 2]{\label{fig:ROC2}\includegraphics[width=7cm]{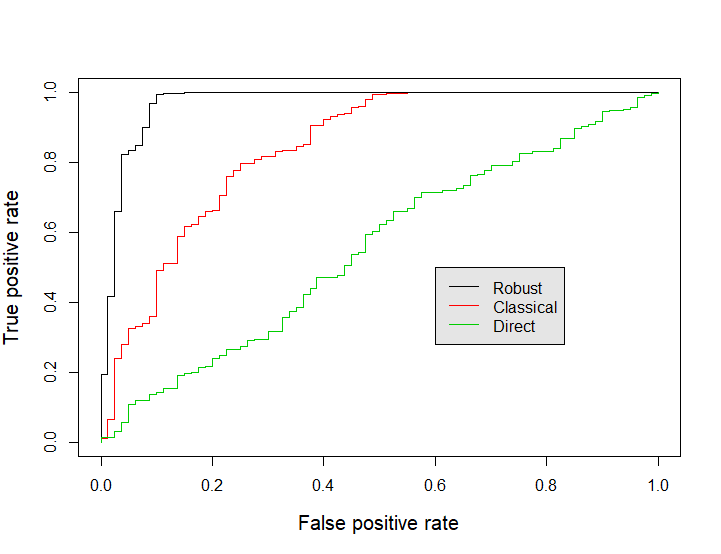}} 
\caption{ROC curve for one instance of Scenario 1 and 2 with the proportion of outlier a$=0.2$ and proportion trimmed $\alpha=0.2$.}\label{fig:ROC_plots}
\end{figure}
    
    \begin{table}
\begin{center}
  \caption{ \quad Average AUC values over 100 replications for Scenario 1, using proportion of outliers a$=0.1,0.2$ and 0.3. Using Direct compared to classic FPCA with BIC, and using robust FPCA with RBIC and trimming levels $\alpha=0.1,0.2$ and 0.3.}
\begin{tabular}{@{}l *{13}{D{.}{.}{4}}@{}}
\multicolumn{4}{l}{} & & \\
\toprule
\multicolumn{1}{l}{} & \multicolumn{2}{c}{} & \multicolumn{2}{c}{a=0.1} & \multicolumn{2}{c}{a=0.2} & \multicolumn{2}{c}{a=0.3} \\
\midrule
\multicolumn{1}{l}{Direct} & \multicolumn{2}{c}{-} & \multicolumn{2}{c}{0.532} & \multicolumn{2}{c}{0.538} & \multicolumn{2}{c}{0.550}\\
\midrule
\multicolumn{1}{l}{Classic}  & \multicolumn{2}{c}{$\alpha=0.0$} & \multicolumn{2}{c}{0.960} & \multicolumn{2}{c}{0.898} & \multicolumn{2}{c}{0.797}\\
\midrule
\multicolumn{1}{l}{Robust} & \multicolumn{2}{c}{$\alpha=0.1$} & \multicolumn{2}{c}{0.995} & \multicolumn{2}{c}{0.991} & \multicolumn{2}{c}{0.953} \\
\multicolumn{1}{l}{} & \multicolumn{2}{c}{$\alpha=0.2$} & \multicolumn{2}{c}{\textbf{0.996}} & \multicolumn{2}{c}{\textbf{0.996}} & \multicolumn{2}{c}{0.987} \\
\multicolumn{1}{l}{} & \multicolumn{2}{c}{$\alpha=0.3$} & \multicolumn{2}{c}{\textbf{0.996}} & \multicolumn{2}{c}{\textbf{0.996}} & \multicolumn{2}{c}{\textbf{0.990}} \\
\bottomrule
\end{tabular} 
\label{table:scenario1_AUC}
\end{center}
\end{table}

    \begin{table}
\begin{center}
  \caption{ \quad Average AUC values over 100 replications for Scenario 2, using proportion of outliers a$=0.1,0.2$ and 0.3. Using Direct compared to classic FPCA with BIC, and using robust FPCA with RBIC and trimming levels $\alpha=0.1,0.2$ and 0.3.}
\begin{tabular}{@{}l *{13}{D{.}{.}{4}}@{}}
\multicolumn{4}{l}{} & & \\
\toprule
\multicolumn{1}{l}{} & \multicolumn{2}{c}{} & \multicolumn{2}{c}{a=0.1} & \multicolumn{2}{c}{a=0.2} & \multicolumn{2}{c}{a=0.3} \\
\midrule
\multicolumn{1}{l}{Direct} & \multicolumn{2}{c}{-} & \multicolumn{2}{c}{0.512} & \multicolumn{2}{c}{0.548} & \multicolumn{2}{c}{0.554} \\
\midrule
\multicolumn{1}{l}{Classic}  & \multicolumn{2}{c}{$\alpha=0.0$} & \multicolumn{2}{c}{0.922} & \multicolumn{2}{c}{0.838} & \multicolumn{2}{c}{0.734}\\
\midrule
\multicolumn{1}{l}{Robust} & \multicolumn{2}{c}{$\alpha=0.1$} & \multicolumn{2}{c}{\textbf{0.985}} & \multicolumn{2}{c}{0.964} & \multicolumn{2}{c}{0.932} \\
\multicolumn{1}{l}{} & \multicolumn{2}{c}{$\alpha=0.2$} & \multicolumn{2}{c}{0.980} & \multicolumn{2}{c}{\textbf{0.980}} & \multicolumn{2}{c}{0.966} \\
\multicolumn{1}{l}{} & \multicolumn{2}{c}{$\alpha=0.3$} & \multicolumn{2}{c}{0.980} & \multicolumn{2}{c}{\textbf{0.980}} & \multicolumn{2}{c}{\textbf{0.968}} \\
\bottomrule
\end{tabular} 
\label{table:scenario2_AUC}
\end{center}
\end{table}

\section{Jet Engine data}
\label{sec:engine}

The Jet engine dataset contains sensor measurements taken during 199 Vibration Survey (VS) manoeuvres. This manoeuvre has a distinctive shape with a slow acceleration and a slow deceleration, with examples shown in Figure \ref{fig:wave_engine}. We do not have labels for whether any of the individual engines have outliers but we do have log books from the engine test, from which we can obtain insights into the Vibration Survey manoeuvres which our method flags as outliers. There are a number of temperature features measured within an engine including the TGT, discussed previously. In addition we have four other temperature readings T25, T30, TCAR and TCAF, from sensors measuring temperature in different parts of the engine. All the temperature features are shown in Figure \ref{fig:outliers_temp_all}. The TCAR is particularly interesting as it has two distinct curve behaviours. It is also worth noting that the temperature values are distinctively higher at the end of the manoeuvre than at the beginning even though the engine speeds are the same. This highlights the trajectory-dependent behaviour that we seek to model. The VS manoeuvres time series are of similar length. To standardise we have fitted a B-spline basis of 400 basis functions to each to ensure the time series are well approximated. Then we have taken 1000 equally spaced points on the B-spline representations to be our inputs $x_i(t)$ and $y_i(t)$.   

We will be applying the outlier detection algorithm described in Algorithm \ref{alg:OD}, which uses RFLR. We will compare these outliers with those detected on the temperature curves directly and using CFLR and BIC in Algorithm \ref{alg:OD}. We can look at the residuals curves to determine if the outliers do indeed look abnormal. In particular we want to show that using functional regression we are able to determine outliers that would otherwise be missed by investigating the temperature curves directly.   

Using the depth based outlier detection (Direct) \citep{Bande2008} directly on the temperature curves (with a default threshold of $\delta=0.01$), we obtain the outliers in Table \ref{table:trent_outliers}. We can see that the outliers in the TPR are the same as the outliers in the temperature features. This suggests the outliers being identified are arising from the controller induced variability. We therefore need to model the dependency between the control feature (TPR) and the temperature features.

We applied the outlier detection algorithm given in Algorithm \ref{alg:OD} using CFLR and BIC with threshold $\delta=0.01$. The outliers identified are given in Table \ref{table:trent_outliers}. The residuals curves are shown in Figure \ref{fig:outliers_temp_class}, with the outliers coloured in blue. It is not clear from this plot that the outliers are truly different from the other data.


Lastly we applied Algorithm \ref{alg:OD} using RFLR and RBIC with threshold $\delta=0.01$. The outlier samples are given in Table \ref{table:trent_outliers} for each temperature feature. In Figure \ref{fig:outliers_temp_more} we have the residual curves using RFLR. We can see that the RFLR model fits the majority of the temperature curves well. The outliers that are picked up clearly look abnormal, with significant deviations from the general behaviour. The RFLR model is therefore able to identify interesting behaviour, which may otherwise have been undetected. Engineers have informed us that Sample 24 comes from an engine in which they detected damaged hardware. All the other outliers in the RFLR column of Table \ref{table:trent_outliers} were also noted to come from engines that displayed odd behaviour during the Pass-Off test. This is not the case for the outliers reported in the CFLR column.

In Figure \ref{fig:outliers_temp_all} we have a plot of the temperature parameters with the outliers identified using the curves directly in green, those using the RFLR model in red and those detected by both in purple. We can see that the outliers from the RFLR model do not necessarily appear as abnormal if we look at the temperature curves directly. Sample 106 is identified as an outlier by multiple temperature features and also when the depth based outlier detection is used on the temperature curves directly. Comparing the outliers identified using a classical approach, we can see Sample 24 is identified as an outlier multiple times using the classical and robust approaches. However most of the outliers from the classical approaches differ with the outliers identified using the robust approach. We can also see that the outliers using the RFLR are significantly more distinctive than the outliers using CFLR. 

\begin{figure}
\centering     
\subfigure[TPR]{\label{fig:all_out_tpr}\includegraphics[width=6cm]{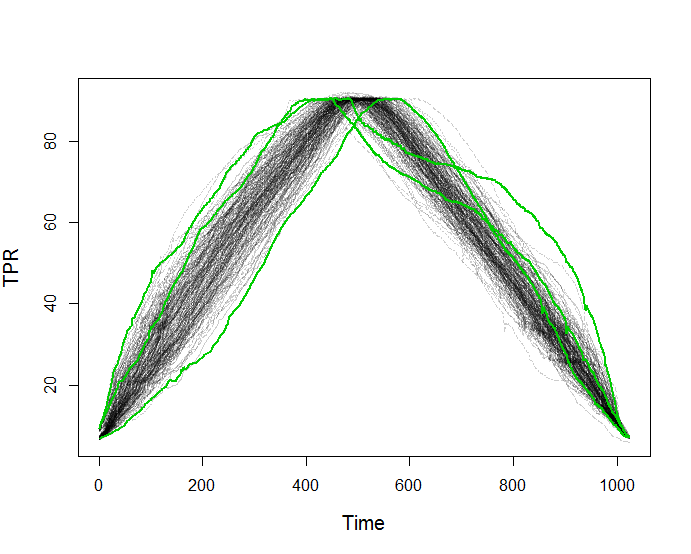}}
\subfigure[T25]{\label{fig:all_out_t25}\includegraphics[width=6cm]{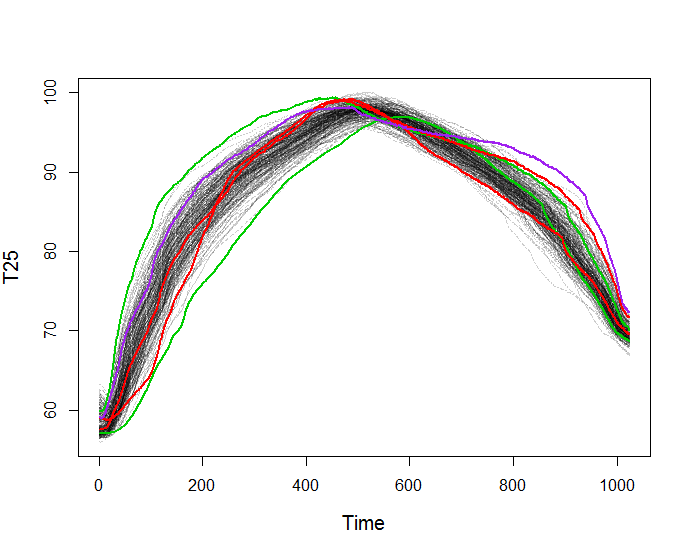}}
\subfigure[T30]{\label{fig:all_out_t30}\includegraphics[width=6cm]{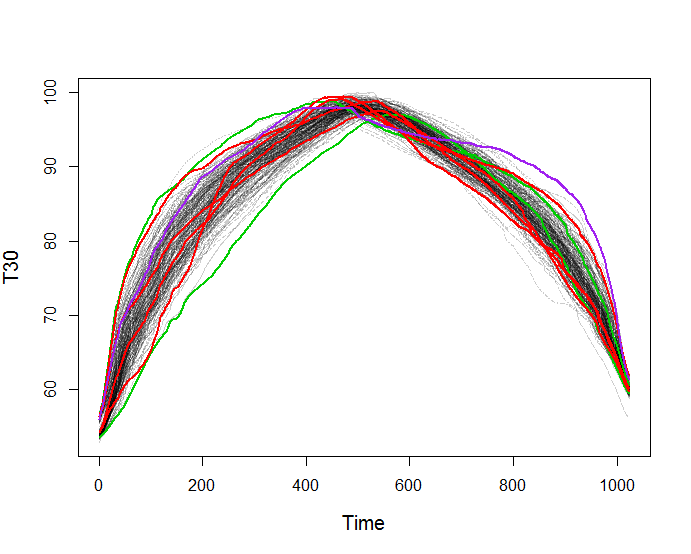}}
\subfigure[TGT]{\label{fig:all_out_tgt}\includegraphics[width=6cm]{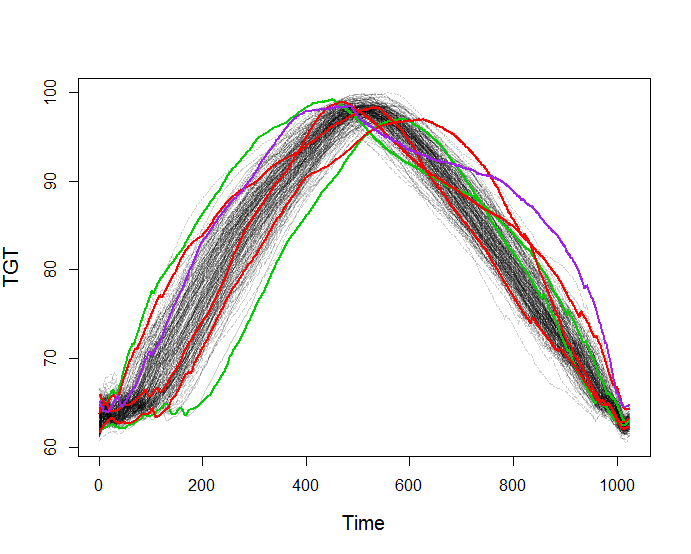}}
\subfigure[TCAR]{\label{fig:all_out_tcar}\includegraphics[width=6cm]{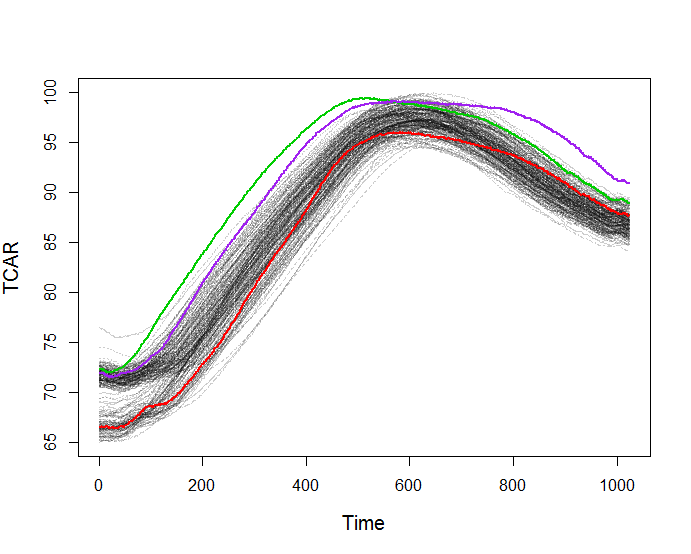}}
\subfigure[TCAF]{\label{fig:out_out_tcaf}\includegraphics[width=6cm]{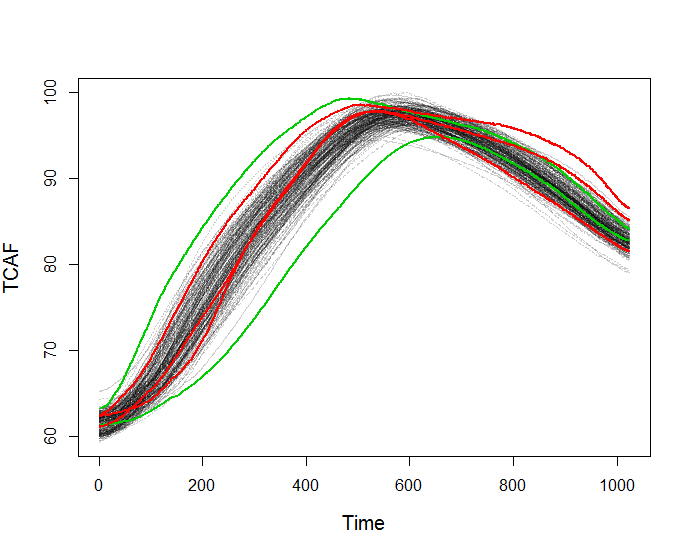}}
\caption{Plots of the TPR, T25, T30, TGT, TCAR and TCAF time series with outliers using robust FLR in red; those using the curves directly in green and those for both in purple.}\label{fig:outliers_temp_all}
\end{figure}

\begin{figure}
\centering     
\subfigure[T25]{\label{fig:all_class_t25}\includegraphics[width=6cm]{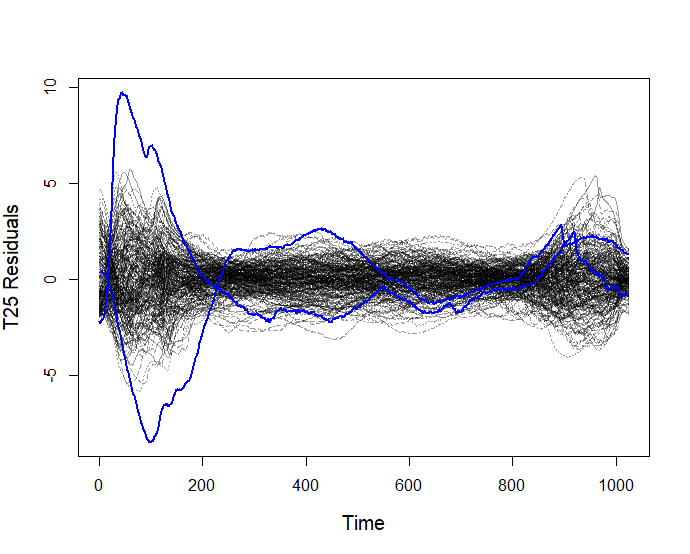}}
\subfigure[T30]{\label{fig:all_class_t30}\includegraphics[width=6cm]{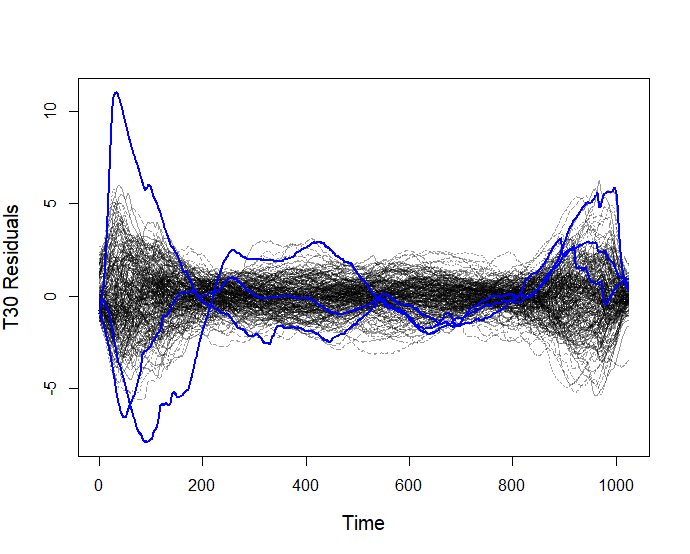}}
\subfigure[TGT]{\label{fig:all_class_tgt}\includegraphics[width=6cm]{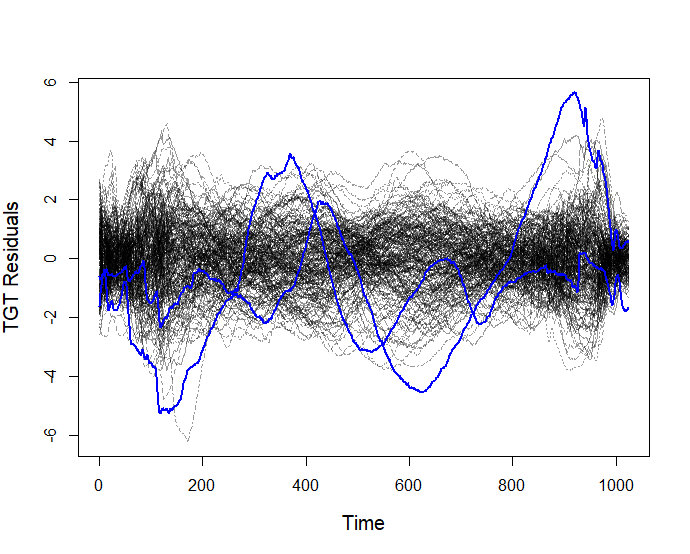}}
\subfigure[TCAR]{\label{fig:all_class_tcar}\includegraphics[width=6cm]{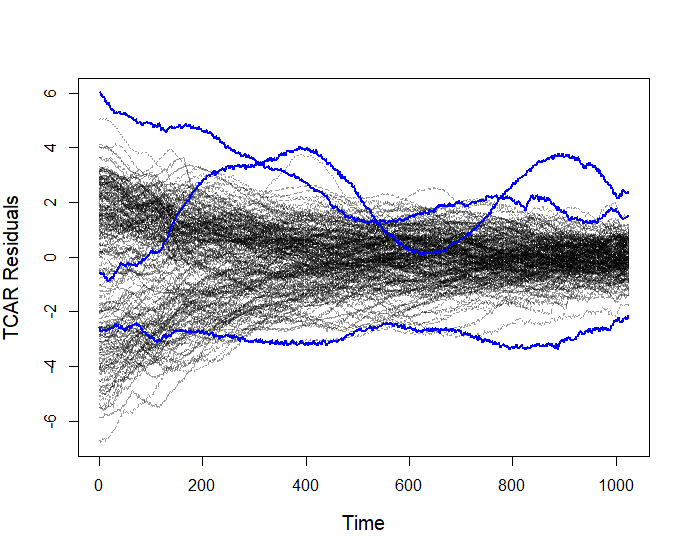}}
\subfigure[TCAF]{\label{fig:all_class_tcaf}\includegraphics[width=6cm]{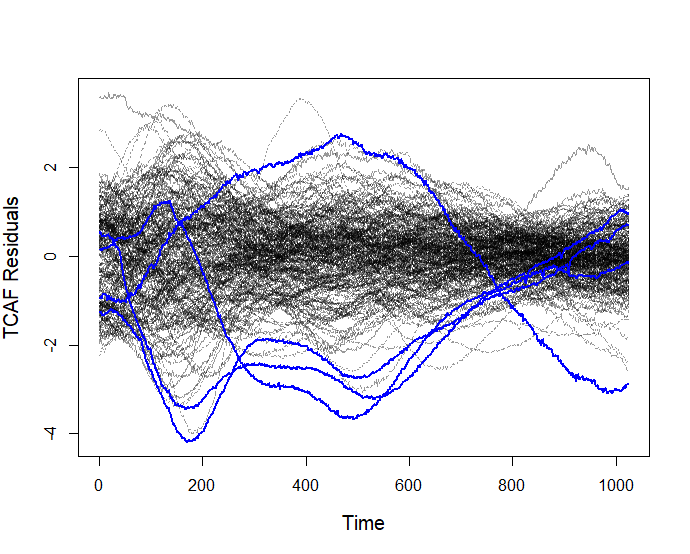}}
\caption{Plots of the residuals of the T25, T30, TGT, TCAR and TCAF with outliers using classical FLR in blue.}\label{fig:outliers_temp_class}
\end{figure}

\begin{figure}
\centering     
\subfigure[T25]{\label{fig:all_more_t25}\includegraphics[width=6cm]{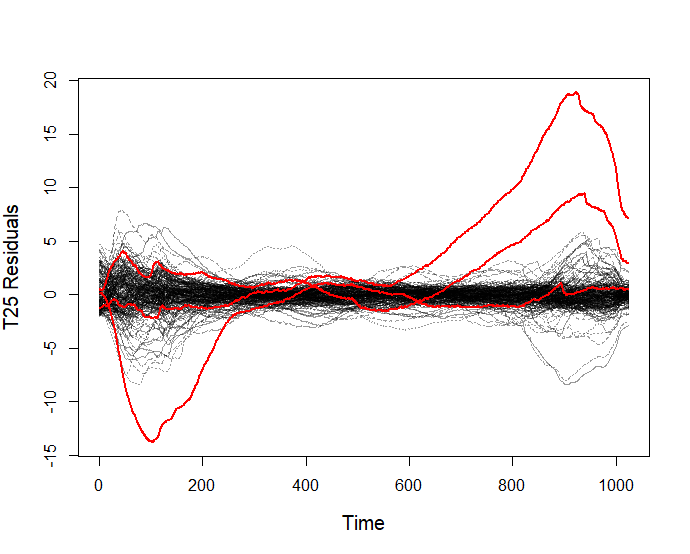}}
\subfigure[T30]{\label{fig:all_more_t30}\includegraphics[width=6cm]{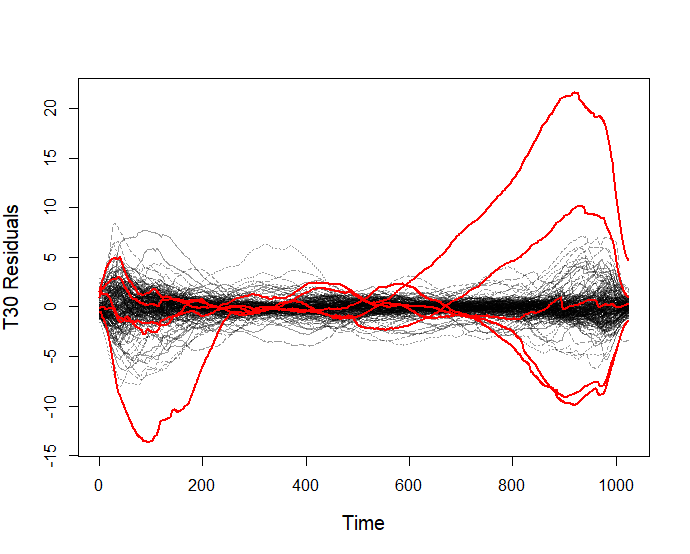}}
\subfigure[TGT]{\label{fig:all_more_tgt}\includegraphics[width=6cm]{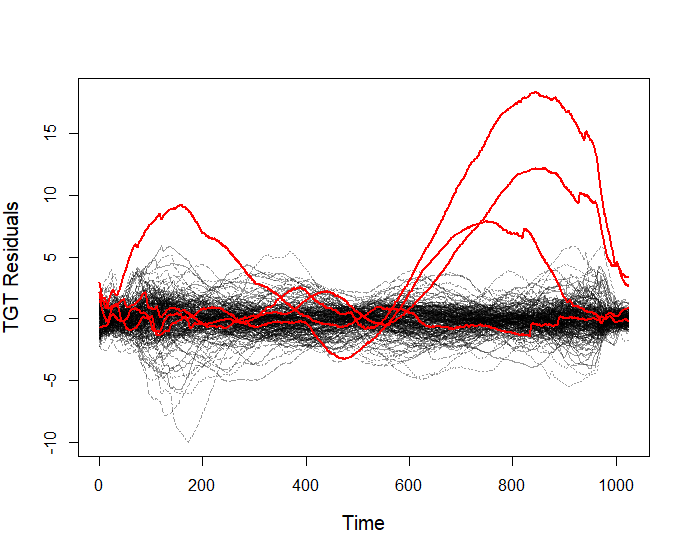}}
\subfigure[TCAR]{\label{fig:all_more_tcar}\includegraphics[width=6cm]{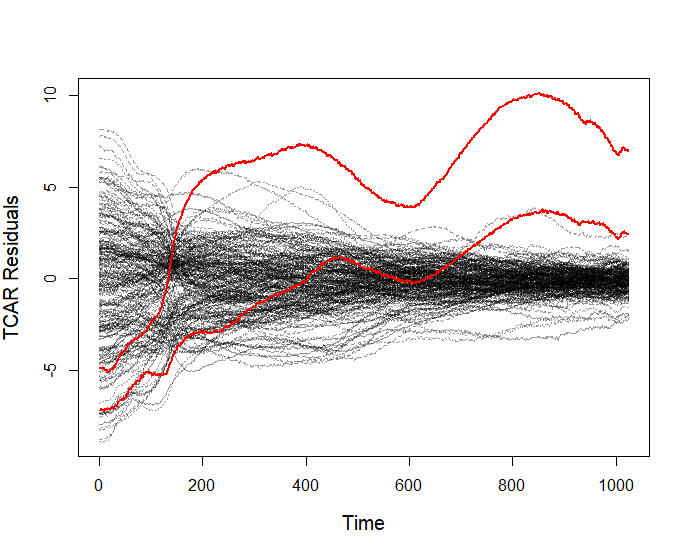}}
\subfigure[TCAF]{\label{fig:all_more_tcaf}\includegraphics[width=6cm]{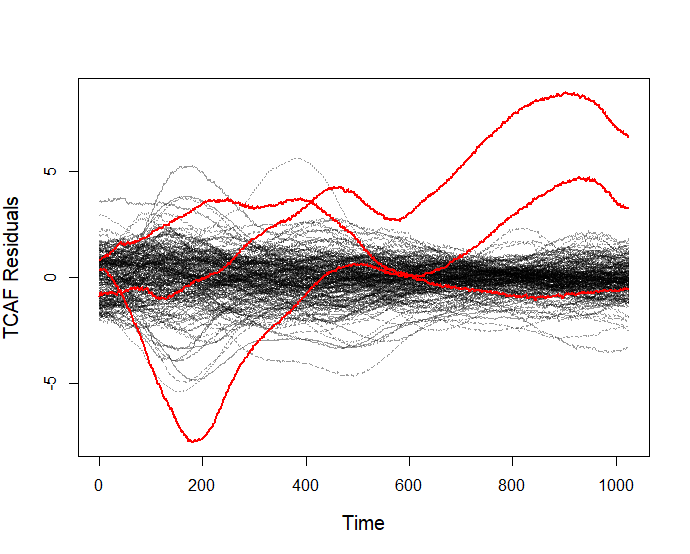}}
\caption{Plots of the residuals of the T25, T30, TGT, TCAR and TCAF with outliers using robust FLR in red.}\label{fig:outliers_temp_more}
\end{figure}

\begin{table}[]
\begin{center}
\begin{tabular}{|l|l|l|l|}
\hline
\textbf{Temp} & \textbf{Direct} & \textbf{CFLR} & \textbf{RFLR}    \\ \hline
TPR           & 33, 106, 167     & -             & -                \\ \hline
T25           & 33, 106, 167     & 24, 182        & 24, 70, 106        \\ \hline
T30           & 33, 106, 167     & 24, 182, 192        & 24, 44, 70, 106, 196 \\ \hline
TGT           & 33, 106, 167     & 119, 153 & 44, 70, 106, 117    \\ \hline
TCAR          & 33, 106         & 36, 91, 106     & 70, 106           \\ \hline
TCAF          & 33, 167         & 65, 167, 170, 171   & 24, 70, 106        \\ \hline
\end{tabular}
\end{center}
\caption{Outliers detected for temperature features (Temp) using outlier detection on the temperature features directly (Direct), and the outliers found using CFLR and RFLR.}
\label{table:trent_outliers}
\end{table}



\section{Conclusion}
\label{sec:conc}

There exist a number of functional regression models for functional inputs and responses, however these methods are not robust to outliers. We have introduced a robust FLR model that is able to produce good model fits in the presence of outliers. Alongside the robust FLR model we have also introduced a robust model selection procedure and proven the consistency of the robust FLR and model selection procedure. Using a simulation study we have shown the need for a robust approach to obtain good models in the presence of outliers. The robust FLR model is also effective in identifying global and localised outliers. Finally using jet engine sensor data as a motivating application for robust FLR we have identified unusual temperature behaviour. In particular the outliers identified in the jet engine sensor data would not have been detected if we modelled the response variables independently of human controlled driving variable.

\bigskip
\begin{center}
{\large\bf SUPPLEMENTARY MATERIAL}
\end{center}

\begin{description}


\item[RobFLR:] A zip file of the R code used for simulation study

\end{description}

\bibliographystyle{Chicago}

\bibliography{Bibliography-MM-MC}

\begin{thebibliography}{}

\bibitem[\protect\citeauthoryear{Agull\'{o}, Croux, and Van~Aelst}{Agull\'{o}
  et~al.}{2008}]{Agullo2008}
Agull\'{o}, J., C.~Croux, and S.~Van~Aelst (2008).
\newblock The multivariate least-trimmed squares estimator.
\newblock {\em J. Multivar. Anal.\/}~{\em 99\/}(3), 311--338.

\bibitem[\protect\citeauthoryear{Arribas-Gil and Romo}{Arribas-Gil and
  Romo}{2014}]{ArribasGil2014}
Arribas-Gil, A. and J.~Romo (2014).
\newblock Shape outlier detection and visualization for functional data: the
  outliergram.
\newblock {\em Biostatistics\/}~{\em 15 4}, 603--19.

\bibitem[\protect\citeauthoryear{Bali, Boente, Tyler, and Wang}{Bali
  et~al.}{2011}]{bali2011}
Bali, J.~L., G.~Boente, D.~E. Tyler, and J.-L. Wang (2011).
\newblock Robust functional principal components: A projection-pursuit
  approach.
\newblock {\em Ann. Statist.\/}~{\em 39\/}(6), 2852--2882.

\bibitem[\protect\citeauthoryear{Chiou, Yang, and Chen}{Chiou
  et~al.}{2016}]{CHIOU2016}
Chiou, J.-M., Y.-F. Yang, and Y.-T. Chen (2016).
\newblock Multivariate functional linear regression and prediction.
\newblock {\em Journal of Multivariate Analysis\/}~{\em 146}, 301 -- 312.
\newblock Special Issue on Statistical Models and Methods for High or Infinite
  Dimensional Spaces.

\bibitem[\protect\citeauthoryear{Croux and Ruiz-Gazen}{Croux and
  Ruiz-Gazen}{1996}]{Croux1996}
Croux, C. and A.~Ruiz-Gazen (1996).
\newblock A fast algorithm for robust principal components based on projection
  pursuit.
\newblock In A.~Prat (Ed.), {\em COMPSTAT}, Heidelberg, pp.\  211--216.
  Physica-Verlag HD.

\bibitem[\protect\citeauthoryear{Cuevas, Febrero, and Fraiman}{Cuevas
  et~al.}{2007}]{Cuevas2007}
Cuevas, A., M.~Febrero, and R.~Fraiman (2007).
\newblock Robust estimation and classification for functional data via
  projection-based depth notions.
\newblock {\em Comput. Stat.\/}~{\em 22\/}(3), 481--496.

\bibitem[\protect\citeauthoryear{Dai and Genton}{Dai and
  Genton}{2018}]{Dai2018}
Dai, W. and M.~G. Genton (2018).
\newblock Multivariate functional data visualization and outlier detection.
\newblock {\em Journal of Computational and Graphical Statistics\/}~{\em
  27\/}(4), 923--934.

\bibitem[\protect\citeauthoryear{Eilers and Marx}{Eilers and
  Marx}{1996}]{eilers1996}
Eilers, P. H.~C. and B.~D. Marx (1996).
\newblock Flexible smoothing with b -splines and penalties.
\newblock {\em Statist. Sci.\/}~{\em 11\/}(2), 89--121.

\bibitem[\protect\citeauthoryear{Febrero-Bande, Galeano, and
  Gonzãlez-Manteiga}{Febrero-Bande et~al.}{2008}]{Bande2008}
Febrero-Bande, M., P.~Galeano, and W.~Gonzãlez-Manteiga (2008).
\newblock Outlier detection in functional data by depth measures, with
  application to identify abnormal nox levels.
\newblock ~{\em 19}, 331 -- 345.

\bibitem[\protect\citeauthoryear{Ferraty, Keilegom, and Vieu}{Ferraty
  et~al.}{2012}]{FERRATY2012}
Ferraty, F., I.~V. Keilegom, and P.~Vieu (2012).
\newblock Regression when both response and predictor are functions.
\newblock {\em Journal of Multivariate Analysis\/}~{\em 109}, 10 -- 28.

\bibitem[\protect\citeauthoryear{Hubert, Rousseeuw, and Segaert}{Hubert
  et~al.}{2015}]{Hubert2015}
Hubert, M., P.~J. Rousseeuw, and P.~Segaert (2015).
\newblock Multivariate functional outlier detection.
\newblock {\em Statistical Methods {\&} Applications\/}~{\em 24\/}(2),
  177--202.

\bibitem[\protect\citeauthoryear{Ivanescu, Staicu, Scheipl, and
  Greven}{Ivanescu et~al.}{2015}]{Ivanescu2015}
Ivanescu, A.~E., A.-M. Staicu, F.~Scheipl, and S.~Greven (2015).
\newblock Penalized function-on-function regression.
\newblock {\em Computational Statistics\/}~{\em 30\/}(2), 539--568.

\bibitem[\protect\citeauthoryear{Kalogridis and Aelst}{Kalogridis and
  Aelst}{2019}]{Kalogridis2018}
Kalogridis, I. and S.~V. Aelst (2019).
\newblock Robust functional regression based on principal components.
\newblock {\em Journal of Multivariate Analysis\/}~{\em 173}, 393 -- 415.

\bibitem[\protect\citeauthoryear{Machado}{Machado}{1993}]{Machado1993}
Machado, J. A.~F. (1993).
\newblock Robust model selection and {M}-estimation.
\newblock {\em Econometric Theory\/}~{\em 9\/}(3), 478--493.

\bibitem[\protect\citeauthoryear{{Matsui}}{{Matsui}}{2017}]{matsui2017}
{Matsui}, H. (2017).
\newblock {Quadratic regression for functional response models}.
\newblock {\em ArXiv e-prints\/}.

\bibitem[\protect\citeauthoryear{{Morris}}{{Morris}}{2015}]{Morris2015}
{Morris}, J.~S. (2015).
\newblock {Functional Regression}.
\newblock {\em Annual Review of Statistics and Its Application\/}~{\em 2},
  321--359.

\bibitem[\protect\citeauthoryear{Pollard}{Pollard}{2012}]{pollard2012}
Pollard, D. (2012).
\newblock {\em Convergence of stochastic processes}.
\newblock Springer Science \& Business Media.

\bibitem[\protect\citeauthoryear{Ramsay and Silverman}{Ramsay and
  Silverman}{2005}]{Ramsay2005}
Ramsay, J. and B.~W. Silverman (2005).
\newblock {\em {Functional Data Analysis (Springer Series in Statistics)}}.

\bibitem[\protect\citeauthoryear{Ramsay and Dalzell}{Ramsay and
  Dalzell}{1991}]{tools}
Ramsay, J.~O. and C.~J. Dalzell (1991).
\newblock Some tools for functional data analysis.
\newblock {\em Journal of the Royal Statistical Society. Series B
  (Methodological)\/}~{\em 53\/}(3), 539--572.

\bibitem[\protect\citeauthoryear{Rousseeuw, Raymaekers, and Hubert}{Rousseeuw
  et~al.}{2018}]{Rousseeuw2018}
Rousseeuw, P.~J., J.~Raymaekers, and M.~Hubert (2018).
\newblock A measure of directional outlyingness with applications to image data
  and video.
\newblock {\em Journal of Computational and Graphical Statistics\/}~{\em
  27\/}(2), 345--359.

\bibitem[\protect\citeauthoryear{Scheipl, Staicu, and Greven}{Scheipl
  et~al.}{2015}]{Scheipl2015}
Scheipl, F., A.-M. Staicu, and S.~Greven (2015).
\newblock Functional additive mixed models.
\newblock {\em Journal of Computational and Graphical Statistics\/}~{\em
  24\/}(2), 477--501.

\bibitem[\protect\citeauthoryear{Shang}{Shang}{2014}]{Shang2014}
Shang, H.~L. (2014).
\newblock A survey of functional principal component analysis.
\newblock {\em AStA Advances in Statistical Analysis\/}~{\em 98\/}(2),
  121--142.

\bibitem[\protect\citeauthoryear{Sun and Genton}{Sun and
  Genton}{2011}]{Sun2011}
Sun, Y. and M.~G. Genton (2011).
\newblock Functional boxplots.
\newblock {\em Journal of Computational and Graphical Statistics\/}~{\em
  20\/}(2), 316--334.

\bibitem[\protect\citeauthoryear{Yao, Müller, and Wang}{Yao
  et~al.}{2005}]{yaofang2005}
Yao, F., H.-G. Müller, and J.-L. Wang (2005).
\newblock Functional linear regression analysis for longitudinal data.
\newblock {\em Ann. Statist.\/}~{\em 33\/}(6), 2873--2903.

\end{thebibliography}
\end{document}